\documentclass{article}
\usepackage{amsthm, amsmath,amsfonts,graphicx, amssymb,bm,geometry, mathtools}
\usepackage[utf8]{inputenc} 

\usepackage{enumitem}

\usepackage[T1]{fontenc}    

\usepackage{hyperref,cleveref}
\usepackage[numbers]{natbib}

\usepackage{algorithm}
\usepackage{algorithmic}
\usepackage{thm-restate}

\usepackage{thm-restate}
\theoremstyle{plain}
\newtheorem{theorem}{Theorem}[section]
\newtheorem{corollary}[theorem]{Corollary}

\newtheorem{proposition}[theorem]{Proposition}

\newtheorem{assumption}[theorem]{Assumption}

\theoremstyle{definition}
\newtheorem{definition}[theorem]{Definition}
\newtheorem{remark}[theorem]{Remark}
\newtheorem{example}[theorem]{Example}

\DeclareMathOperator*{\argmin}{arg\,min}

\usepackage{xcolor}

\title{Cooperation in Threshold Public Projects with Binary Actions\footnote{A short version of this paper is accepted by IJCAI'21.  Yiling Chen and Fang-Yi Yu are supported by the National Science Foundation under grants IIS 2007887}}
\author{Yiling Chen\thanks{Harvard University \texttt{yiling@seas.harvard.edu}}\and Biaoshuai Tao\thanks{John Hopcroft Center for Computer Science, Shanghai Jiao Tong University \texttt{bstao@sjtu.edu.cn}} \and Fang-Yi Yu\thanks{Harvard University \texttt{fangyiyu@seas.harvard.edu}}}
\date{}

\begin{document}

\maketitle

\begin{abstract}
When can cooperation arise from self-interested decisions in public goods games? And how can we help agents to act cooperatively? We examine these classical questions in a \emph{pivotal participation game}, a variant of public good games, where heterogeneous agents make binary participation decisions on contributing their endowments, and the public project succeeds when it has enough contributions. 

We prove it is NP-complete to decide the existence of a cooperative Nash equilibrium such that the project succeeds. 
We also identify two natural special scenarios where this decision problem is tractable.  
We then propose two algorithms to help cooperation in the game. Our first algorithm adds an external investment to the public project, and our second algorithm uses matching funds.  
We show that the cost to induce a cooperative Nash equilibrium is near-optimal for both algorithms.  Finally, the cost of matching funds can always be smaller than the cost of adding an external investment. Intuitively, matching funds provide a greater incentive for cooperation than adding an external investment does.
\end{abstract}

\section{Introduction}

Consider a town that seeks donations of collectibles from its residents with the goal to open a public museum if enough collectibles are received, or countries deciding whether to participate in an environmental agreement (e.g. Kyoto Protocol)~\citep{kaul1999global,finus2002game,battaglini2016participation} to collectively reduce greenhouse gas emission. These settings have the nature of public goods games---the non-exclusiveness of a public good may lead to free ride but the success of the project requires a certain level of participation. Can cooperative behavior arise from self-interested decisions to lead to success in the project? How can we help agents to act cooperatively? 

We propose a variant of public goods games, \emph{pivotal participation game}, to answer these two classic questions. Each agent has an endowment and decides whether to contribute to a public project. If the total contribution is enough for the project to succeed, everyone regardless of their contribution receives rewards from the project.  

Our pivotal participation game has two characteristics. First, participants make binary decisions on participation (contributing their endowment or not). This allows us to capture the indivisibility of contributions in collectible donation and participation agreement. Restriction to binary actions is also suitable for other public goods domains such as vaccination~\citep{brito1991externalities}, adoption of innovations~\citep{DYBVIG1983231}, public health insurance~\citep{goldstein1976group}, or even practicing social distancing~\citep{cato2020social}. Second, participants have heterogeneous endowments and rewards. Many well-studied public goods games are symmetric among agents: agents have symmetric endowments~\citep{van1983minimal,palfrey1984participation} and the reward is uniformly distributed among agents~\citep{rapoport1988provision} or only depends on the number of contributing agents~\citep{hirshleifer1983weakest,jackson2010social}.  In our setting, agents' endowments (value of collectibles, cost of reducing greenhouse gas emissions) are heterogeneous, and their reward from the project's success (values of visiting museum or benefits from clean air) can also be different.

We first prove that it is NP-complete to decide whether there exists a Nash equilibrium such that the project succeeds (Theorem~\ref{thm:npc}). This result suggests large groups may struggle to cooperate in public projects, which is supported in the literature.~\citep{olson2009logic}  Then we show that successful cooperation is always possible in two natural scenarios: 1) agents' endowments are small or the reward levels are large enough in \Cref{sec:ratio} and 2) the minimal successful cooperation has the same number of agents in \Cref{sec:balanced}

Then we develop algorithms to help achieving cooperation.  We consider two natural interventions: The algorithms can 1) add an external investment $\delta$ to the public project in \Cref{sec:external}, or 2) commit to a matching fund so that for every $1$ value of investment from agents, the algorithms add $\rho$ to the project in \Cref{sec:matching}.
Both interventions are common in practice.~\citep{baker2009matching,karlan2020can}  We want to minimize the additional cost due to those interventions.  We first show that it is NP-hard to approximate the optimal cost under both interventions to within any finite multiplicative factors.  Then, we design two algorithms with near-optimal cost in the additive sense.  (Theorem~\ref{thm:approximation_external} and \ref{thm:approximation_matchingfunds}) Finally, we compare these two interventions and prove that matching funds can promote cooperation with no more cost than adding an external investment can. (Proposition~\ref{prop:compare})   This provides a theoretical explanation for the efficacy of matching funds in the field.~\citep{karlan2007does}

\paragraph{Related works.}
Our work considers the computational issue of searching Nash equilibria of public goods games, and designs algorithms to help cooperation.  Our model falls in the category of threshold public goods with binary actions.  However, for most of the existing models in this category, computational issue has seldom been considered and successful cooperation can be found efficiently.  One pioneering work ~\citep{palfrey1984participation}  considers participation games where agents have homogeneous unit endowment. Thus, participation by any subset of sufficient number of agents is an equilibrium.  \citeauthor{rapoport1988provision} [\citeyear{rapoport1988provision}] and \citeauthor{bolle2014class} [\citeyear{bolle2014class}] consider that agents' rewards only depend on whether the project succeeds or not, while ours also depends on the project's total value of contributions.  As a result, successful cooperation can be computed efficiently in their games.  Another important category of public goods games is network public goods game~\citep{bramoulle2007public}.  For this category, recent works also consider the computational complexity of finding equilibria and design algorithms to help cooperation~\citep{Yu_Zhou_Brantingham_Vorobeychik_2020,kempe2020inducing}.  However, the main interest of their games are network structures instead of heterogeneity of agents' endowments and reward levels.

\section{Model}\label{sec:model}
We consider the \emph{pivotal participation game} where agents have heterogeneous endowment and binary actions deciding whether to invest all her endowment or none. 
Formally, a group of $n$ agents wants to fund a project which has a \emph{threshold} $\tau\in \mathbb{R}_{\ge0}$ for success.  Each agent $i\in [n]$ has an indivisible \emph{endowment} that has an equivalent value of $e_i\in \mathbb{R}_{\ge 0}$, and she can decide whether to participate by contributing all her endowment or free ride, $a_i \in \{0,1\}$.  Thus, agents' strategy profile can be denoted by $S := \{i:a_i = 1\}\subseteq [n]$ which is the set of agents who participate.  Let $e(S):= \sum_{j\in S} e_j$ be the total value of contributions made by agents in $S$.  The project succeeds if $e(S)\ge \tau$.
When the project succeeds, each agent's reward depends on her \emph{reward level} $0<m_i<1$.  Specifically, given a strategy profile $S$, agent $i$'s utility is

\begin{equation}\label{eqn:utility}
    U_i(S) = e_i\mathbf{1}[i\notin S]+m_i e(S)\mathbf{1}\left[e(S)\ge \tau\right].
\end{equation}
We use $(\vec{e}, \vec{m}, \tau)$ to denote a pivotal participation game instance, where $\vec{e} = (e_1, \ldots, e_n)$ encodes each agent's endowment and $\vec{m} = (m_1, \ldots, m_n)$ encodes each agent's reward level.  
A strategy profile $S$ is a \emph{Nash equilibrium} if no one would unilaterally deviate: If $i$ invests, i.e., $i\in S$, $U_i(S)\ge U_i(S\setminus\{i\})$.  If $i\notin S$, $U_i(S)\ge U_i(S\cup\{i\})$.  
A strategy profile $S$ is a \emph{cooperative Nash equilibrium} if it is a Nash equilibrium and the project succeeds,  $e(S)\ge \tau$.

\section{Best Response and Nash Equilibria}
First note that, fixing other agents' actions, an agent $i$ will free ride if the project succeeds without her investment or cannot succeed with her investment.  To see this, let $S_{-i} = \{j\neq i:a_j = 1\}$ be other agents' action. If agent $i$'s action does not affect whether the project succeeds,
$U_i(S_{-i}\cup\{i\})-U_i(S_{-i}) \le -e_i+m_i e_i<0$.
That is, agent $i$ will invest only if her investment can change the project from unsuccessful to successful, $e(S_{-i})<\tau\le  e(S_{-i})+e_i$.
However, the following example shows this is just a necessary condition.
\begin{example}\label{ex:big}
Let $\vec{e}=(2,2,2,9)$, $\vec{m}=(0.2,0.2,0.2,0.5)$, and $\tau = 10$.  The project's success requires agent $4$ to participate.  However, agent $4$ will not participate, because even if everyone participates, the total investment is $15$ and agent $4$'s reward is at most $15m_4 = 7.5$ which is smaller than her endowment $9$.  
\end{example}
The above example shows that even with enough endowments the agents cannot cooperate due to their heterogeneity.

The following proposition shows that the best response depends on not only whether the project succeeds but also how large her endowment $e_i$ and reward level $m_i$ are. 
\begin{proposition}\label{prop:conditional}
Given the other agents' actions $S_{-i}$, agent $i$'s best response is participating if and only if 
\begin{align}
    &e(S_{-i})<\tau\le  e(S_{-i})+e_i, \text{ and}\label{eq:condition1}\\
    &\frac{1-m_i}{m_i}e_i\le e(S_{-i}).\label{eq:condition2}
\end{align}
\end{proposition}
\begin{proof}
If both \eqref{eq:condition1} and \eqref{eq:condition2} are true,
\begin{align*}
    &U_i(S_{-i}\cup\{i\})\\
    =& m_ie(S_{-i}\cup\{i\})\mathbf{1}[e(S_{-i}\cup\{i\})\geq\tau]\tag{by Eqn.~\eqref{eqn:utility}}\\
    =& m_i e(S_{-i}\cup\{i\})\tag{by the second part of Eqn.~\eqref{eq:condition1}}\\
    =& m_i e(S_{-i}) +m_i e_i\\
    \ge& (1-m_i)e_i +m_i e_i\tag{by Eqn.~\eqref{eq:condition2}}\\
    =& U_i(S_{-i})\tag{by Eqn.~\eqref{eq:condition1}}
\end{align*}
Therefore, agent $i$'s best response is participating.

Conversely, suppose Eqn.~\eqref{eq:condition1} does not hold.  
Agent $i$'s action does not affect whether the project succeeds.
By Eqn.~\eqref{eqn:utility}, $U_i(S_{-i}\cup\{i\})-U_i(S_{-i}) \le -e_i+m_i e_i<0$.  Thus, agent $i$'s best response is not participating.  Finally, if Eqn.~\eqref{eq:condition2} does not hold, 
\begin{align*}
    U_i(S_{-i}\cup\{i\})\le& m_i e(S_{-i})+m_i e_i\tag{may not succeed}\\
    <& (1-m_i)e_i+m_i e_i\tag{Eqn.~\eqref{eq:condition2} is false}\\
    =& U_i(S_{-i}), 
\end{align*}
so agent $i$'s best response is not participating.  This completes the proof.
\end{proof}

With Proposition~\ref{prop:conditional}, we can show that an agent will not participate if her endowment is very large.
\begin{corollary}\label{cor:big}
If  $e_i\ge\tau$, agent $i$ will not participate.
\end{corollary}
\begin{proof}
Suppose $e_i\ge \tau$ and agent $i$ participates.  Other agent will free ride because the project already succeeds.  However, agent $i$ does not want to participate by himself/herself, because $m_i<1$.
\end{proof}

\subsection{Nash Equilibria}
By Proposition~\ref{prop:conditional}, we have the following characterization of a successful cooperation $S$ which is a Nash equilibrium that makes the project successful.
\begin{proposition}\label{prop:NE}
Given a pivotal game $(\vec{e}, \vec{m}, \tau)$, $S$ is a cooperative Nash equilibrium if and only if for all $i\in S$
\begin{equation}\label{eq:NE}
    \max\left\{\tau, \frac{e_i}{m_i}\right\}\le e(S)< \tau +e_i.
\end{equation}
\end{proposition}

\Cref{fig:my_label} shows that social assortment can emerge when agents' endowments are heterogeneous. Agents with large endowments do not cooperate with small endowment agents and form a Nash equilibrium.  

\begin{figure}
    \centering
    \includegraphics[width = 0.65\textwidth]{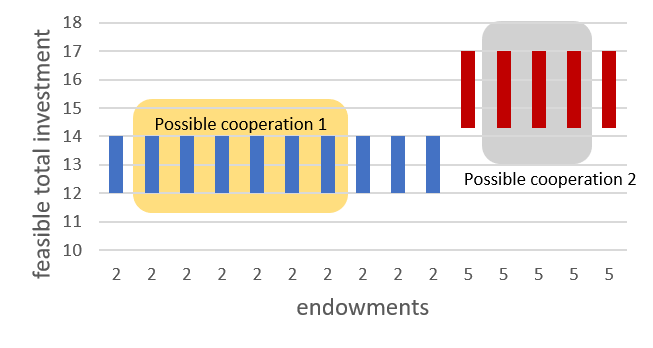}
    \caption{Let $\tau = 12$.  There are $15$ agents whose reward level is $0.35$. $10$ agents have endowment $2$ and the other agents have endowment $5$.  With \Cref{prop:NE}, for an agent with endowment $2$ to invest, the sum of investment need to be between $\max\{12, 2/0.35\} = 12$ and $12+2 = 14$ which is represented as a blue bar in the figure.  However, when the endowment is $5$, the total investment needs to be between $14.3$ and $17$  which is represented as a red bar.\\
    There are several possible cooperation.  Any collection of $6$ agents with endowment $2$ is one (yellow region), and any collection of $3$ agents with endowment $5$ is another (grey region). 
    However, by Proposition~\ref{prop:NE}, there does not exist a cooperative Nash equilibrium where both agents with endowments $2$ and $5$ invest.}
    \label{fig:my_label}
\end{figure}

Further, we show that deciding if there is a cooperative Nash equilibrium is NP-complete.  
This implies that searching such an cooperation among a large group of agents is computationally difficult.

\begin{theorem}\label{thm:npc}
Deciding if a pivotal participation game  $(\vec{e}, \vec{m}, \tau)$ has a cooperative Nash equilibrium is NP-complete.
\end{theorem}
\begin{proof}
This decision problem is clearly in NP: the set of agents $S\neq\emptyset$ that invest can be served as a certificate, and Proposition~\ref{prop:NE} can decide if $S$ is a Nash equilibrium in polynomial time.
It remains to show that the problem is NP-hard.

We reduce the problem from a well-known NP-complete problem, the partition problem \citep{garey1979computers}: given a set $\mathcal{S}$ of $2T$ positive integers $c_1,\ldots,c_{2T}$, decide if there exists a subset $\mathcal{T}$ of $T$ integers whose sum is exactly $\frac12\sum_{i=1}^{2T}c_i$.
We assume without loss of generality that $c_1\leq\cdots\leq c_{2T}$.

We construct a pivotal participation game instance $(\vec{e}, \vec{m}, \tau)$ as follows.
Let $M=100\sum_{i=1}^{2T}c_i$ and $N=100MT$.
The instance contains $n=2T+2$ agents, with $e_i=N+M+2c_i$ for $i=1,\ldots,2T$, and $e_{2T+1}=N+1$ and $e_{2T+2}=N$.
It is straightforward to see that $e_{2T+2}<e_{2T+1}<e_1\leq\cdots\leq e_{2T}$.
The parameter $\tau$ is set to $\tau=1+e_{2T+1}+\frac12\sum_{i=1}^{2T}e_i=\frac12\sum_{i=1}^{2T}2c_i+(N+M)T+N+2$.
Finally, $\vec{m}$ is set to the followings: for each $i=1,\ldots,n$, $m_i=\frac{e_i}{\tau+e_{2T+2}-1}$.
For each $i=1,\ldots,n$, we have $\max\{\frac{e_i}{m_i},\tau\}=\tau+e_{2T+2}-1=\frac12\sum_{i=1}^{2T}2c_i+(N+M)T+2N+1$, which must be a lower bound to $e(S)$ by Proposition~\ref{prop:NE}.

Suppose the partition instance is a YES instance, and let $\mathcal{T}$ be the subset of $T$ integers such that $\sum_{i\in\mathcal{T}}c_i=\frac12\sum_{i=1}^{2T}c_i$.
By slightly abuse of the notation, let $\mathcal{T}$ be the subset of the first $2T$ agents that corresponds to the set of $T$ selected integers.
We show that $S=\mathcal{T}\cup\{2T+1,2T+2\}$ is a valid Nash equilibrium.
To see this, the total investment is $e(S)=2\cdot\left(\sum_{i\in\mathcal{T}}c_i\right)+(M+N)T+2N+1=2\cdot\left(\frac12\sum_{i=1}^{2T}c_i\right)+(M+N)T+2N+1=\tau+e_{2T+2}-1=\max\{\tau,\frac{e_i}{m_i}\}$ for each $i=1,\ldots,n$.
Moreover, for each $i=1,\ldots,n$, by noticing that $e_{2T+2}\leq e_i$, we have $e(S)\leq\tau+e_i-1<\tau+e_i$.
By Proposition~\ref{prop:NE}, $S$ is a Nash equilibrium.

Suppose the partition instance is a NO instance.
We aim to show that the only Nash equilibrium is $S=\emptyset$.
To derive a contradiction, suppose $S\neq\emptyset$ is a Nash equilibrium, and we consider two cases: $2T+2\in S$ and $2T+2\notin S$.

Suppose $2T+2\in S$. By Proposition~\ref{prop:NE}, we must have $\frac{e_{2T+2}}{m_{2T+2}}=\tau+e_{2T+2}-1\leq e(S)<\tau + e_{2T+2}$.
Since all $e_i$'s are integers, this implies $e(S)=\tau+e_{2T+2}-1$, and it must be that $\sum_{i\in S\setminus\{2T+2\}}e_i=\tau-1=e_{2T+1}+\frac12\sum_{i=1}^{2T}e_i$.
By our construction, each of $e_1,\ldots,e_{2T}$ is an even number, while $e_{2T+1}$ is an odd number.
Thus, $e_{2T+1}+\frac12\sum_{i=1}^{2T}e_i$ is an odd number (note that $\sum_{i=1}^{2T}e_i$ is a multiple of $4$).
By the parity of the equation, we must have $2T+1\in S$.
This implies $\sum_{i\in S\setminus\{2T+1,2T+2\}}e_i=\frac12\sum_{i=1}^{2T}e_i$.
Let $\mathcal{T}=S\setminus\{2T+1,2T+2\}$.
We have $\sum_{i\in \mathcal{T}}e_i=\frac12\sum_{i=1}^{2T}e_i$.

Next, it is easy to see that $|\mathcal{T}|=T$.
If $|\mathcal{T}|<T$, $\sum_{i\in \mathcal{T}}e_i\leq(M+N)(T-1)+\sum_{i\in\mathcal{T}}2c_i<(M+N)(T-1)+\frac1{50}M<(M+N)T<\frac12\sum_{i=1}^{2T}e_i$.
If $|\mathcal{T}|>T$, $\sum_{i\in \mathcal{T}}e_i>(M+N)(T+1)>(M+N)T+\frac1{50}M>\frac12\sum_{i=1}^{2T}e_i$.

Knowing that $|\mathcal{T}|=T$, by canceling the $(M+N)$ terms on both sides of the equation $\sum_{i\in \mathcal{T}}e_i=\frac12\sum_{i=1}^{2T}e_i$ and then divide the both sides by $2$, we have $\sum_{i\in\mathcal{T}}c_i=\frac12\sum_{i=1}^{2T}c_i$, which contradicts to that the partition instance is a NO instance.

It remains to derive a contradiction for the case $2T+2\notin S$.
Firstly, we show that $|S|\geq T+2$.
If $|S|\leq T+1$, we have 
$e(S)=\sum_{i\in S}e_i<(N+M)(T+1)+\frac1{50}M+1<(M+N)T+1.1N<\tau+e_{2T+2}-1=\frac{e_i}{m_i}$, which implies $S$ cannot be a Nash equilibrium by Proposition~\ref{prop:NE}.

Having shown that $|S|\geq T+2$, we discuss two sub-cases: $2T+1\notin S$ and $2T+1\in S$.
For the first sub-case $2T+1\notin S$, since $e_i>M+N$ for all $i=1,\ldots,2T$, we have $e(S)=\sum_{i\in S}e_i>(T+2)(M+N)$.
On the other hand, for each $i=1,\ldots,2T$, we have $\tau+e_i=\frac12\sum_{i=1}^{2T}2c_i+(M+N)T+N+2+(M+N)+c_i<(M+N)(T+1)+N+\frac1{25}M+2<(M+N)(T+2)$.
Therefore, $e(S)>\tau+e_i$ for each $i=1,\ldots,2T$.
Since $S\subseteq\{1,\ldots,2T\}$, $S$ cannot be a Nash equilibrium by Proposition~\ref{prop:NE}.

Consider the second sub-case $2T+1\in S$.
There are at least $T+1$ agents in $\{1,\ldots,2T\}$, so $e(S)\geq(M+N)(T+1)+N+1=(M+N)T+2N+M+1$.
On the other hand, $\tau+e_{2T+1}=\frac12\sum_{i=1}^{2T}2c_i+(M+N)T+2N+3<(M+N)T+2N+M<e(S)$.
By Proposition~\ref{prop:NE}, $S$ cannot be a Nash equilibrium, and we derive a contradiction again.

We have derive contradictions for all possible cases for $S\neq\emptyset$.
Therefore, the only Nash equilibrium is $S=\emptyset$ when the partition instance is a NO instance.
We conclude the theorem.
\end{proof}

Although deciding the existence of cooperative equilibria is NP-complete, we will see two special scenarios where this decision problem is tractable.

\subsubsection{Low Endowment-to-Reward Ratio}\label{sec:ratio}
In the following theorem, we show that a cooperative Nash equilibrium always exists and can be computed efficiently when the ratio of endowment to reward level is smaller than $\tau$ for each agents (as long as the sum of all agents' endowments exceeds $\tau$). Intuitively, our result shows that cooperation becomes easy, if agents' endowments are small or the reward level are large enough, in contrast to Example~\ref{ex:big}.

\begin{theorem}\label{thm:greedy}
Given a pivotal participation game $(\vec{e}, \vec{m}, \tau)$, if $e_i/m_i\le \tau$ for all $i$, there exists a successful cooperation if and only if $\sum_{i=1}^n e_i\ge \tau$.
Moreover, we can find a cooperative Nash equilibrium $S$ in polynomial time.
\end{theorem}
We will show that every minimal subset of agents $S$ (defined below) is a Nash equilibrium.

\begin{definition}\label{def:minimal}
Given a pivotal participation game $(\vec{e}, \vec{m}, \tau)$, a subset of agents $S$ is \emph{minimal} if
\begin{enumerate}
    \item $e(S)\geq \tau$, and
    \item for any $j\in S$, $e(S\setminus\{j\}) < \tau$.
\end{enumerate}
\end{definition}

\begin{proposition}\label{prop:greedy}
Given a pivotal participation game $(\vec{e}, \vec{m}, \tau)$, if $e_i/m_i\le\tau$ for all $i$, every minimal subset of agents $S$ is a Nash equilibrium.
\end{proposition}
\begin{proof}
We prove this by showing that every minimal subset of agents $S$ satisfies Eqn.~(\ref{eq:NE}) in Proposition~\ref{prop:NE}.
Given an arbitrarily minimal subset of agents $S$, since $e_i/\tau\leq m_i$, the left part of (\ref{eq:NE}) holds due to 1 of Definition~\ref{def:minimal}.
Suppose the right part of (\ref{eq:NE}) fails for some agent $j\in S$.
We have $e(S)\geq\tau+e_j$, which implies $e(S\setminus\{j\})\geq\tau$, which contradicts to 2 of Definition~\ref{def:minimal}.
Thus, the right part of (\ref{eq:NE}) also holds for any $i\in S$.
\end{proof}

Finally, Theorem~\ref{thm:greedy} can be shown by constructing a minimal subset of agents $S$.
\begin{proof}[Proof of Theorem~\ref{thm:greedy}]
Suppose $\sum_{i=1}^ne_i\geq\tau$.
The following simple algorithm finds a minimal subset of agents $S$: initialize $S=\emptyset$; iteratively add an arbitrary agent $i$ to $S$, and terminate the algorithm when $e(S)\geq\tau$.
The algorithm will terminates with $e(S)\geq\tau$ since $\sum_{i=1}^ne_i\geq\tau$.
It is easy to see that the algorithm finds a minimal subset of agents $S$, and by Proposition~\ref{prop:greedy}, $S$ is a Nash equilibrium.

Suppose $\sum_{i=1}^ne_i<\tau$. The left part of (\ref{eq:NE}) in Proposition~\ref{prop:NE} will always fail.
Therefore, there does not exist a cooperative Nash equilibrium.
\end{proof}

\subsubsection{Balanced Instances}\label{sec:balanced}
We consider another special case where the success of the project always requires the same number of agents, and show that the problem of deciding if a cooperative Nash equilibrium exists can be solved in polynomial time.

\begin{definition}
A pivotal participation game instance $(\vec{e},\vec{m},\tau)$ is \emph{balanced} if all the minimal subsets of agents (see Definition~\ref{def:minimal}) have the same cardinality.
\end{definition}

For a typical natural example of balanced instances, we can have agents' endowments being approximately the same, say, agents' endowments are approximately $N$.
The instance is balanced as long as $\tau$ is not around any $bN$ for some integer $b\in\mathbb{Z}^+$.
For example, if agents' endowments are approximately $N$ and $\tau=3.5N$, the success of the project always needs exactly $4$ agents' endowments. 

\begin{theorem}
Deciding the existence of a cooperative Nash equilibrium is in P for balanced pivotal participation game instances.
\end{theorem}

\begin{algorithm}[htb]
\caption{Decide the existence of cooperative Nash equilibrium for balanced instances}
\label{alg:balanced}
\textbf{Input}: a balanced pivotal participation game instance $(\vec{e}, \vec{m}, \tau)$\\
\textbf{Output}: a set of investing agents $S$, or $S=\emptyset$ (no cooperative Nash equilibrium)
\begin{algorithmic}[1] 
\STATE compute $l_i=\max\{\tau,\frac{e_i}{m_i}\}$ for each $i\in[n]$
\STATE let $m$ be the cardinality of a minimal subset of agents
\STATE sort the agents such that $e_1\geq e_2\geq\cdots\geq e_n$
\STATE initialize $S=\{1,\ldots,m\}$ and $T=\{m+1,\ldots,n\}$
\WHILE{there exists $i\in S$ with $l_i>e(S)$}
\IF{$T=\emptyset$}
\STATE \textbf{return} $S=\emptyset$
\ENDIF
\STATE remove $i$ from $S$
\STATE add the first element $j$ of $T$ to $S$
\STATE remove $j$ from $T$
\ENDWHILE
\STATE \textbf{return} $S$
\end{algorithmic}
\end{algorithm}

Firstly, Algorithm~\ref{alg:balanced} runs in polynomial time.
To see this, the while-loop is executed at most $n-m$ times, as each iteration of the while-loop removes one element from the set $T$, which initially has cardinality $n-m$.

To show the correctness of the algorithm, we consider the two cases regarding if Line~13 is reached.
Suppose Line~13 is reached.
We know that $|S|=m$ (as the cardinality of $S$ is never changed throughout the while-loop) and $l_i\leq e(S)$ for all $i\in S$.
Thus, the left inequality of Eqn.~\eqref{eq:NE} in Proposition~\ref{prop:NE} is satisfied for all $i\in S$.
Suppose the right inequality $e(S)<\tau+e_i$ fails for some $i\in S$, we have $e(S)\geq\tau+e_i$ and $e(S\setminus\{i\})\geq\tau$, so $S$ is not minimal, which contradicts to the balanced assumption.
Thus, the right inequality is also satisfied for all $i\in S$, and $S$ is a Nash equilibrium.

Suppose Line~13 is never reached, and the algorithm terminates at Line~7.
We prove that each element $i$ removed from $S$ can never be a part of a Nash equilibrium.
This can be proved by induction.

Consider the first iteration of the while-loop.
If $i\in S$ satisfies $l_i>e(S)$, for any set $S'$ of $m$ agents, we have $e(S')\leq e(S)<l_i$ since $S$ contains the $m$ agents with the largest endowments by Line~3.
Therefore, any $S'$ with $i\in S'$ will violate Eqn.~\eqref{eq:NE}, so $i$ cannot be a part of any Nash equilibrium.

For the inductive step, suppose each of the previously removed elements cannot be a part of a Nash equilibrium.
At the current iteration of the while-loop, we know that a Nash equilibrium must be a subset of $S\cup T$, and $S$ is a subset of $m$ agents having largest endowments.
By the same arguments, if $i\in S$ satisfies $l_i>e(S)$, for any set $S'\subseteq S\cup T$ of $m$ agents, we have $e(S')\leq e(S)<l_i$ since $S$ contains the $m$ agents with the largest endowments.
Therefore, any $S'$ with $i\in S'$ will violate Eqn.~\eqref{eq:NE}, so $i$ cannot be a part of any Nash equilibrium.

By the time Line~7 is reached, we have removed $n-m+1$ agents each of which cannot be a part of a Nash equilibrium.
The remaining $m-1$ agents cannot form a Nash equilibrium by the balanced assumption.
Thus, we know that a cooperative Nash equilibrium does not exist.

\section{Interventions}
Here we consider two possible interventions to promote agents to cooperate. We can \emph{add an external investment} or \emph{commit to a matching fund}.  We call an intervention \emph{valid} if there is a cooperative Nash equilibrium after the the intervention.  In both interventions, We want to find a valid intervention that minimizes the additional cost.   

We first show that it is NP-hard to approximate the optimal cost within any finite factor under both interventions (Theorem~\ref{thm:hardness-of-approximation}).  Then, we design two algorithms (\Cref{alg:external_investment,alg:matchingFunds}) which take parameters of the game $(\vec{e}, \vec{m}, \tau)$ and output a valid intervention.  We show the cost of each algorithm's intervention is close to the optimal with small additive errors (\Cref{thm:approximation_external,thm:approximation_matchingfunds}).  Finally, we compare these two interventions and prove that matching fund is more powerful than adding an  external investment in Proposition~\ref{prop:compare}.

\subsection{External Investment}\label{sec:external}
In this section, we consider mechanisms that add investment $\delta$ to the public project in addition to agents' investment.  Thus, agent $i$'s utility, $U_i(S)$ becomes
\begin{equation}
    e_i\mathbf{1}[i\notin S]+m_i (e(S)+\delta)\mathbf{1}\left[e(S)+\delta\ge \tau\right].
\end{equation}

By a similar analysis, we can show the following proposition similar to Proposition~\ref{prop:NE}.
We leave the details to the readers.

\begin{proposition}\label{prop:NE_investment}
Given a pivotal participation game instance $(\vec{e}, \vec{m}, \tau)$ with external investment $\delta$, $S$ is a cooperative Nash equilibrium if and only if for all $i\in S$
\begin{equation}\label{eq:NE_investment}
    \max\left\{\tau, \frac{e_i}{m_i}\right\}\le e(S)+\delta< \tau +e_i.
\end{equation}
\end{proposition}

Thereafter, we will assume $\tau+e_i>\frac{e_i}{m_i}$ for each agent $i$, for otherwise, $i$ will never be a part of a cooperative Nash equilibrium by Proposition~\ref{prop:NE_investment}, and we can remove $i$ from consideration.

We say a solution $(S,\delta)$ is \emph{valid} if $S$ is a Nash equilibrium under external investment $\delta$.
The following theorem shows that the optimal solution to this minimization problem always exists.

\begin{restatable}{theorem}{thmexistexternal}\label{thm:exist_external}
Given a pivotal participation game $(\vec{e}, \vec{m}, \tau)$, let $\Delta\subseteq[0,\infty)$ be the set of all valid external investment $\delta$.
We have $\inf\Delta\in\Delta$.
\end{restatable}
\begin{proof}
Firstly, we have $\Delta\neq\emptyset$.
Since we have assumed $\tau+e_i>\frac{e_i}{m_i}$ for every $i$, $(S=\{1\}, \delta=\max\{\tau,\frac{e_1}{m_1}\}-e_1)$ is always a valid solution.

Given a subset of agents $S$, let $\Delta(S)$ be the set of all possible values $\delta$ such that $S$ is a Nash equilibrium under external investment $\delta$.
Set $\Delta(\emptyset)=\emptyset$.
Since the power set of $[n]$ is finite, it is easy to prove that 
$$\inf\Delta=\min_{S:S\subseteq[n],\Delta(S)\neq\emptyset}\inf\Delta(S).$$
It suffices to show that either $\inf\Delta(S)\in\Delta(S)$ or $\Delta(S)=\emptyset$ for each $S$.
We consider a fixed $S\neq\emptyset$ from now on.

For each $i\in[n]$, let $l_i=\max\{\tau,e_i/m_i\}$ and $u_i=\tau+e_i$.

If $\max_{i\in S}l_i\geq\min_{i\in S}u_i$ or $e(S)\geq\min_{i\in S}u_i$, we know immediately $\Delta(S)=\emptyset$ by Proposition~\ref{prop:NE_investment}.
Suppose $\max_{i\in S}l_i<\min_{i\in S}u_i$ and $e(S)<\min_{i\in S}u_i$ from now on.

If $\max_{i\in S}l_i\leq e(S)<\min_{i\in S}u_i$, by Proposition~\ref{prop:NE_investment}, $\delta=0$ makes $S$ a Nash equilibrium, and $\inf\Delta(S)=0\in\Delta(S)$.
If $e(S)<\max_{i\in S}l_i$, it is easy to see that $\delta=\max_{i\in S}l_i-e(S)$ is the minimum $\delta$ to make $S$ a Nash equilibrium, and $\inf\Delta(S)=\max_{i\in S}l_i-e(S)\in\Delta(S)$.
\end{proof}

Firstly, Theorem~\ref{thm:npc} straightforwardly implies that the problem is NP-hard to approximate to any finite multiplicative factor. 

\begin{theorem}\label{thm:hardness-of-approximation}
Given a pivotal participation game instance $(\vec{e}, \vec{m}, \tau)$, letting $\delta^\ast\geq0$ be the minimum valid external investment, for any $F>0$ that may depend on $(\vec{e}, \vec{m}, \tau)$, it is NP-hard to approximate $\delta^\ast$ to within factor $F$.
\end{theorem}
\begin{proof}
Theorem~\ref{thm:npc} implies that it is NP-complete to decide if $\delta^\ast=0$.
This theorem concludes immediately.
\end{proof}

Fortunately, this minimization problem has an additive term approximation algorithm, with the additive term being the maximum endowments among the agents: $e=\max_{i:1\leq i\leq n}\{e_i\}$.  The algorithm is described in Algorithm~\ref{alg:external_investment}.

\begin{algorithm}[htb]
\caption{Approximation algorithm for pivotal participation game with external investment}
\label{alg:external_investment}
\textbf{Input}: a pivotal participation game instance $(\vec{e}, \vec{m}, \tau)$\\
\textbf{Output}: a set of agent $S$ that invest; the amount of external investment $\delta$
\begin{algorithmic}[1] 
\STATE for each $i\in [n]$, let $u_i=\tau +e_i$ and $l_i=\max\{\tau, \frac{e_i}{m_i}\}$
\STATE sort agents such that $l_1\leq l_2\leq\cdots\leq l_n$
\FOR{$i=1,\ldots,n$}
\STATE initialize $S_i=\emptyset$
\STATE let $K_i=\{k\in[n]:l_k\leq l_i<u_k\}$ 
\STATE sort agents in $K_i$ so that $e_{k_1}\geq e_{k_2}\geq\cdots\geq e_{k_{|K_i|}}$\label{line:6}
\FOR{$j=1,\ldots,|K_i|$}
\IF{$e(S_i\cup\{k_j\})>l_i$}
\STATE \textbf{break}
\ENDIF
\STATE update $S_i=S_i\cup\{k_j\}$
\ENDFOR
\STATE set $\delta_i=l_i-e(S_i)$
\ENDFOR
\STATE find $i^o=\argmin_i\delta_i$
\STATE \textbf{return} $(S_{i^o},\delta_{i^o})$
\end{algorithmic}
\end{algorithm}

The algorithm works as follows.
It considers $n$ possible values $l_1,\ldots,l_n$ for $e(S)+\delta$.
For each possibility $e(S)+\delta=l_i$, it iteratively adds an agent $k$ with $l_k\leq l_i<u_k$ to $S$ until either no more agent can be added (when the for-loop at Line~7 terminates without reaching Line~9) or $e(S)\leq l_i$ cannot be maintained (when the for-loop at Line~7 exits at Line~9).
The algorithm finds a solution for $e(S)+\delta=l_i$ by setting $\delta=l_i-e(S)$ for $S$ found by the above step.
Finally, the algorithm compares the $n$ solutions and output the optimal one.
It is straightforward to check that the algorithm runs in polynomial time.

\begin{remark}\label{remark:fairness_external}
Our algorithm has a desirable fairness property: ``richer people contributes first''.~\cite{list2009role}
In each of the $n$ solutions the algorithm outputs, Line~\ref{line:6} of the algorithm guarantees that agents with larger endowments invest first, whenever possible.
The larger the index $i$ is, the better the solution $(S_i,\delta_i)$ preserves this property.
In particular, for $i=n$, $(S_n,\delta_n)$ strictly preserves this property (by Line~5 and the fact that $u_k$ is larger whenever $e_k$ is larger).
When implementing Algorithm~\ref{alg:external_investment}, we can compare the $n$ solutions and trade off between the investment $\delta$ and the said fairness property.
\end{remark}

\begin{theorem}\label{thm:approximation_external}
Given a pivotal participation game $(\vec{e}, \vec{m}, \tau)$, letting $\delta^\ast\geq0$ be the minimum valid external investment and $e=\max_{i:1\leq i\leq n}\{e_i\}$, Algorithm~\ref{alg:external_investment} finds a solution $(S,\delta)$ with $\delta\leq\max\{e,\delta^\ast\}$.
\end{theorem}
\begin{proof}
We first show that the solution output by Algorithm~\ref{alg:external_investment} is valid.
We show that the solution $(S_i,\delta_i)$ found in each iteration of the for-loop is valid.
We have $e(S_i)+\delta_i=l_i$ and each $k\in S_i$ satisfies $l_k\leq l_i<u_k$ (guaranteed by Line~13 and 5 of the algorithm), Proposition~\ref{prop:NE_investment} implies that $S_i$ is a Nash equilibrium.
It remains to prove the approximation guarantee $\delta\leq\max\{e,\delta^\ast\}$.

Suppose agents are ordered such that $l_1\leq\cdots\leq l_n$.
Let $(S,\delta)$ be the solution output by the algorithm. 
Let $(S_i,\delta_i)$ be the solution output at $i$-th iteration of the for-loop at Line~3.
Let $(S^\ast,\delta^\ast)$ be the optimal solution.

Firstly, if the inner for-loop at Line~7 exits at Line~9 for certain iteration $i$ of the for-loop at Line~3, then $\delta<e$, in which case the approximation guarantee is proved.
To see this, Line~8 and 13 ensure that $\delta_i= l_i-e(S_i)<e_{k_j}\leq e$.
Since the final solution output by the algorithm is no worse than $(S_i,\delta_i)$, we have $\delta\leq\delta_i<e$.
From now on, we assume that the ``break'' statement at Line~9 has never been executed for all iterations.

Next, we consider the optimal solution $(S^\ast,\delta^\ast)$.
Let $i^\dag$ be the agent with the largest index such that $e(S^\ast)+\delta^\ast\geq l_{i^\dag}$.
We show that each $j\in S^\ast$ satisfies $l_j\leq l_{i^\dag}<u_j$.
Suppose there exists $j\in S^\ast$ with $l_j>l_{i^\dag}$.
We first know $i^\dag<n$.
By Proposition~\ref{prop:NE_investment}, we must have $e(S^\ast)+\delta^\ast\geq l_j\geq l_{i^\dag+1}$ to make the solution valid.
However, this contradicts to our assumption that $i^\dag$ is the agent with the largest index satisfying $e(S^\ast)+\delta^\ast\geq l_{i^\dag}$.
Suppose there exists $j\in S^\ast$ with $u_j\leq l_{i^\dag}$.
Since we assumed $l_{i^\dag}\leq e(S^\ast)+\delta^\ast$, Proposition~\ref{prop:NE_investment} implies that $(S^\ast,\delta^\ast)$ should not have been valid.

In addition, our assumption that the ``break'' statement at Line~9 has never been executed implies that $e(S^\ast)+\delta^\ast=l_{i^\dag}$.
Suppose otherwise $e(S^\ast)+\delta^\ast>l_{i^\dag}$.
We have $\delta^\ast=0$, for otherwise we can always decrease the value of $\delta^\ast$ while maintaining $e(S^\ast)+\delta^\ast\geq l_{i^\dag}$, which contradicts to that $\delta^\ast$ is optimal.
As a result, $e(S^\ast)+\delta^\ast>l_{i^\dag}$ implies $e(S^\ast)>l_{i^\dag}$.
Since we have proved that each $j\in S^\ast$ satisfies $l_j\leq l_{i^\dag}<u_j$, at iteration $i^\dag$ of our algorithm, we must have $S^\ast\subseteq K_{i^\dag}$.
Therefore, $e(K_{i^\dag})\geq e(S^\ast)>l_{i^\dag}$, the ``break'' statement will be executed before all the $|K_{i^\dag}|$ agents are added to $S_{i^\dag}$, which contradicts to our assumption.

Finally, we have proved that $e(S^\ast)+\delta^\ast=l_{i^\dag}$, and we also have $S^\ast\subseteq K_{i^\dag}$ at iteration $i^\dag$ of our algorithm.
Since the ``break'' statement has not been executed, we have $S_{i^\dag}=K_{i^\dag}$, and $\delta_{i^\dag}=l_{i^\dag}-e(S_{i^\dag})\leq l_{i^\dag}-e(S^\ast)=\delta^\ast$.
Therefore, the solution output at $i^\dag$-th iteration of our algorithm is already optimal, and the final output of our algorithm is no worse than this.
\end{proof}

Finally, we provide an examples showing that adding an external investment can sometime hurt cooperation.
\begin{example}[Harm of adding an external investment]\label{ex:harm_external}
Let $\tau = 9$, $\vec{e} = (4,4,4)$ and $\vec{m} = (4/11, 4/11, 4/11)$.  $S=\{1,2,3\}$ is a cooperative Nash equilibrium.
However, if we add an external endowment $\delta = 2$, there is no cooperative Nash equilibrium by Proposition~\ref{prop:NE_investment}.
\end{example}

\subsection{Matching Funds}\label{sec:matching}
In this section, we consider another intervention---matching funds.  
The mechanism can commit to a rate $\rho\ge0$, and increase the investment from $e(S)$ to $(1+\rho)e(S)$. Thus, agent $i$'s utility $U_i(S)$ becomes
\begin{equation}\label{eq:utility_matching}
    e_i\mathbf{1}[i\notin S]+m_i (1+\rho)e(S)\mathbf{1}\left[(1+\rho)e(S)\ge \tau\right].
\end{equation}

\begin{assumption}\label{assumption:public_goods_game}
We assume that the rate of matching funds $\rho$ has a budget constraint $\rho<\bar\rho$ where $\bar\rho=\frac1{\max_{i\in[n]}m_i}-1$.
\end{assumption}

This assumption ensures $(1+\rho)m_i<1$ for each agent $i$.  To justify this assumption, the total reward level $\sum_i m_i$ can be seen as the return on investment of the public project, and is around $1$ in practice, so each agent's reward level is much less than $1$, e.g., $1/n$.  The rate of matching fund $\rho$ is often a small constant (e.g., 1-to-1, 2-to-1, and 3-to-1).  This makes $m_i(1+\rho)$ significantly less than $1$.
Secondly, the behavior of an agent with $m_i(1+\rho)\geq1$ is fundamentally and unnaturally different from an agent in a normal public goods game: from Eqn.~\eqref{eq:utility_matching} and some simple calculations, this agent will always invest as long as the project is successful, even if the project can succeed without him/her!

Similar to Proposition~\ref{prop:NE} and \ref{prop:NE_investment} we have the following characterization of cooperative Nash equilibrium.
The details are left to the readers.
\begin{proposition}\label{prop:NE_matching}
Given a pivotal participation game instance $(\vec{e}, \vec{m}, \tau)$ with matching funds $\rho<\bar\rho$, $S$ is a cooperative Nash equilibrium if and only if for all $i\in S$ 
\begin{equation}\label{eq:NE_matching}
    \max\left\{\tau, \frac{e_i}{m_i}\right\}\le (1+\rho)e(S)< \tau +(1+\rho)e_i.
\end{equation}
\end{proposition}

We also have a minimization problem here.  The objective of this problem can be either formulated by the rate $\rho$ or the cost $\rho\cdot e(S)$ where $S$ is a cooperative Nash equilibrium under \eqref{eq:utility_matching}.
We say a solution $(S,\rho)$ is valid if $S$ is a Nash equilibrium under the rate of matching funds $\rho$.

First, we want to compare the power of these two interventions.  The following result shows that, for any valid external investment $\delta$, there is a valid matching funds with some $\rho$ that has less or equal cost.  Intuitively, this shows that matching funds is more powerful than adding external investment.
\begin{proposition}\label{prop:compare}
For all $(\vec{e}, \vec{m}, \tau)$, for any valid solution $(S,\delta)$ where the external investment $\delta$ is not prohibitively large $\delta<\bar\rho\cdot e(S)$,  $S$ is also a cooperative Nash equilibrium with matching funds $\rho= \delta/e(S)$.  Moreover, the cost of matching funds $\rho\cdot e(S)$ is exactly $\delta$.
\end{proposition}
\begin{proof}
If adding external $\delta$ endowment makes cooperation successful, setting $\rho = \delta/e(S)$, $S$ is also a successful cooperation because of Proposition~\ref{prop:NE_matching} and $\tau+e_i\le \tau+(1+\rho)e_i$ for all $i\in S$.  Moreover, the cost $\rho\cdot e(S)$ is $\delta$ when agents in $S$ cooperate.
\end{proof}

Moreover, there indeed exist games such that the cost of matching fund is strictly less than the cost of external fund.  For example, consider $\tau = 11$, $\Vec{e} = (3, 3, 3)$ and $\Vec{m} = (0.2, 0.2, 0.2)$.  Because the reward levels are too small such that $\frac{e_i}{m_i}=15>\tau+e_i$, no agent will invest by \Cref{prop:NE_investment}.
Thus, the minimal valid external investment is $\delta = 11$.  However, for matching fund, if we set $\rho = 2/3$, $S = \{1, 2, 3\}$ becomes a cooperative equilibrium and the cost is $9\cdot 2/3 = 6$ which is less than the optimal external fund $11$.

The following theorem shows that the solution with optimal rate and the solution with optimal cost always exist, as long as a cooperation is possible.
\begin{restatable}{theorem}{thmexistmatchingfunds}\label{thm:exist_matchingfunds}
Given a pivotal participation game $(\vec{e}, \vec{m}, \tau)$, let $\Delta\subseteq[0,\infty)$ be the set of all valid rates of matching funds $\rho$.
Suppose $\Delta\neq\emptyset$.
We have $\inf\Delta\in\Delta$.
Moreover, there exists a valid solution $(S^\ast,\rho^\ast)$ such that $\rho^\ast\cdot e(S^\ast)\leq\rho\cdot e(S)$ for all possible solutions $(S,\rho)$.
\end{restatable}
\begin{proof}
Suppose $\Delta\neq\emptyset$.
Given a subset of agents $S$, let $\Delta(S)$ be the set of all possible values $\rho$ such that $S$ is a Nash equilibrium under the rate of matching funds $\rho$.
Set $\Delta(\emptyset)=\emptyset$.
Since the power set of $[n]$ is finite, it is easy to prove that 
$$\inf\Delta=\min_{S:S\subseteq[n],\Delta(S)\neq\emptyset}\inf\Delta(S).$$
To show $\inf\Delta\in\Delta$, it suffices to show that either $\inf\Delta(S)\in\Delta(S)$ or $\Delta(S)=\emptyset$ for each $S$.
We consider a fixed $S\neq\emptyset$ from now on.

For each $i\in[n]$, let $l_i=\max\{\tau,e_i/m_i\}$ and $u_i(\rho)=\tau+(1+\rho)e_i$.

If $e(S)\geq\min_{i\in S}u_i(0)$, we have $\Delta(S)=\emptyset$.
To see this, suppose $e(S)\geq u_i(0)$ for agent $i\in S$.
For any $\rho\geq0$, we have 
\begin{align*}
    (1+\rho)e(S)&\geq e(S)+\rho\cdot e_i\tag{since $i\in S$}\\
    &\geq u_i(0)+\rho\cdot e_i=u_i(\rho),
\end{align*}
and Proposition~\ref{prop:NE_matching} implies that $S$ cannot be a Nash equilibrium.
If $\min_{i\in S}u_i(0)>e(S)\geq\max_{i\in S}l_i$, we know $\rho=0$ makes $S$ a Nash equilibrium, and $\inf\Delta(S)=0\in\Delta(S)$.
Suppose $e(S)<\min_{i\in S}u_i(0)$ and $e(S)<\max_{i\in S}l_i$ from now on.

Firstly, Proposition~\ref{prop:NE_matching} implies that each $\rho\in\Delta(S)$ must satisfy $(1+\rho)e(S)\geq\max_{i\in S}l_i$, which is $\rho\geq\frac{\max_{i\in S}l_i}{e(S)}-1$.
Let $\rho^-=\frac{\max_{i\in S}l_i}{e(S)}-1$ be this lower bound of $\Delta(S)$.
If $\rho^-\geq\bar\rho$, we know $\Delta(S)=\emptyset$ immediately, so suppose otherwise.
If $(1+\rho^-)e(S)<\min_{i\in S}u_i(\rho^-)$, we know $\inf\Delta(S)=\rho^-\in\Delta(S)$ by Proposition~\ref{prop:NE_matching}.
If $(1+\rho^-)e(S)\geq\min_{i\in S}u_i(\rho^-)$, then $\Delta(S)=\emptyset$.
To see this, suppose there exists $\rho\in\Delta(S)$.
Let $\rho'=\rho-\rho^-$, and $\rho'\geq0$ since $\rho^-$ is a lower bound.
Let $i\in S$ be an agent such that $(1+\rho^-)e(S)\geq u_i(\rho^-)$.
We have
\begin{align*}
    (1+\rho)e(S)&=(1+\rho^-)e(S)+\rho'e(S)\\
    &\geq u_i(\rho^-)+\rho'e(S)\\
    &\geq u_i(\rho^-)+\rho'e_i\tag{since $i\in S$ and $\rho'\geq0$}\\
    &=u_i(\rho),
\end{align*}
contradicting to $\rho\in\Delta(S)$.
We conclude $\inf\Delta\in\Delta$.

The fact $\inf\Delta(S)\in\Delta(S)$ (for all $S$ with $\Delta(S)\neq\emptyset$) shown just now also implies the existence of $(S^\ast,\rho^\ast)$ in the theorem.
Take
$$S^\ast=\argmin_{S:S\subseteq[n],\Delta(S)\neq\emptyset}\left(\inf\Delta(S)\right)\cdot e(S),$$
and $\rho^\ast=\inf\Delta(S^\ast)$.
We have $\rho^\ast\in\Delta(S^\ast)\subseteq\Delta$.
In addition, for any valid solution $(S,\rho)$, we have $\rho\cdot e(S)\geq(\inf\Delta(S))\cdot e(S)\geq\rho^\ast\cdot e(S^\ast)$.
\end{proof}

Again, Theorem~\ref{thm:npc} straightforwardly implies that the problem is NP-hard to approximate to any finite multiplicative factor for both objectives (the rate $\rho$ and the cost $\rho\cdot e(S)$).

\begin{theorem}\label{thm:hardness-of-approximation2}
Given a pivotal participation game instance $(\vec{e}, \vec{m}, \tau)$, letting $\rho^\ast\geq0$ be the minimum rate for a valid solution $(S^\ast,\rho^\ast)$ and $\rho^\dag\geq0$ be the rate for a valid solution $(S^\dag,\rho^\dag)$ that minimizes the cost $\rho\cdot e(S)$, for any $F>0$ that may depend on $(\vec{e}, \vec{m}, \tau)$, it is NP-hard to approximate $\rho^\ast$ or $\rho^\dag\cdot e(S^\dag)$ to within factor $F$.
\end{theorem}
\begin{proof}
Theorem~\ref{thm:npc} implies that it is NP-complete to decide if there is a valid solution $(S,\rho)$ with $\rho=0$.
This theorem concludes immediately.
\end{proof}

We show that we can achieve additive approximation to both objectives by adaption of Algorithm~\ref{alg:external_investment}, under the very mild assumption $\bar\rho\geq1$ in \Cref{assumption:public_goods_game}.
It is straightforward to check that Algorithm~\ref{alg:matchingFunds} runs in polynomial time.

\begin{algorithm}[htb]
\caption{Approximation algorithm for pivotal participation game with matching funds}
\label{alg:matchingFunds}
\textbf{Input}: a pivotal participation game instance $(\vec{e}, \vec{m}, \tau)$\\
\textbf{Output}: a set of agent $S$ that invest; the rate of matching funds $\rho$
\begin{algorithmic}[1] 
\STATE for each $i\in [n]$, let $l_i=\max\{\tau, \frac{e_i}{m_i}\}$
\STATE compute $\bar\rho=\frac1{\max_{i\in[n]}m_i}-1$
\STATE sort agents such that $l_1\leq l_2\leq\cdots\leq l_n$
\FOR{$i=1,\ldots,n$}
\STATE initialize $S_i=\emptyset$
\STATE let $K_i=\{k:l_k\leq l_i\}$\hfill//different to Algorithm~\ref{alg:external_investment} here
\STATE sort agents in $K_i$ so that $e_{k_1}\geq e_{k_2}\geq\cdots\geq e_{k_{|K_i|}}$
\FOR{$j=1,\ldots,|K_i|$}
\IF{$e(S_i\cup\{k_j\})>l_i$}
\STATE \textbf{break}
\ENDIF
\STATE let $\rho_{\text{temp}}=\frac{l_i}{e(S_i\cup\{k_j\})}-1$
\IF{$l_i\geq\tau+(1+\rho_{\text{temp}})e_{k_j}$}
\STATE \textbf{break}
\ENDIF
\STATE update $S_i=S_i\cup\{k_j\}$, $\rho_i=\rho_{\text{temp}}$
\ENDFOR
\ENDFOR
\STATE find $i^o=\argmin_i\{\rho_i\cdot e(S_i)\}$
\IF{$\rho_{i^o}\geq\bar\rho$}
\STATE \textbf{return} no solution exists
\ENDIF
\STATE \textbf{return} $(S_{i^o},\rho_{i^o})$
\end{algorithmic}
\end{algorithm}

\begin{remark}
 The ``richer people contribute first'' property for Algorithm~\ref{alg:external_investment} mentioned in Remark~\ref{remark:fairness_external} also holds for Algorithm~\ref{alg:matchingFunds}.
\end{remark}

\begin{restatable}{theorem}{thmApproximationMatchingFunds}\label{thm:approximation_matchingfunds}
Given a pivotal participation game $(\vec{e}, \vec{m}, \tau)$ with $\bar\rho\geq1$ such that a cooperative Nash equilibrium exists under certain $\rho\in[0,\bar\rho)$, letting $(S^\ast,\rho^\ast)$ be a valid solution with minimum cost $\rho^\ast\cdot e(S^\ast)$ and $e=\max_{i:1\leq i\leq n}\{e_i\}$, Algorithm~\ref{alg:matchingFunds} finds a collaboration $(S,\rho)$ with cost upper-bounded by $\rho\cdot e(S)\leq\max\{e,\rho^\ast\cdot e(S^\ast)\}$ and rate upper-bounded by $\rho\leq\max\{1,\rho^\ast\}$.
\end{restatable}
\begin{proof}
We first show that the algorithm always outputs valid solutions.
We show this by showing that the solution $(S_i,\rho_i)$ output in each iteration is valid, as long as $\rho_i<\bar\rho$ (the solution with $\rho_i\geq\bar\rho$ will be filtered through Line~19-22).
We first show that $S_i\neq\emptyset$ so that $\rho_i$ has been defined (equivalently, Line~16 has been reached at least once).
For the first element $k_1\in K_i$ that we attempt to add to $S_i$, we have $e_{k_1}<\frac{e_{k_1}}{m_{k_1}}\leq l_{k_1}\leq l_i$, so the if-condition at Line~9 will not be satisfied for the first iteration of the for-loop at Line~8.
Moreover, for the first iteration of the for-loop at Line~8, we have $l_i=(1+\rho_{\text{temp}})e(\emptyset\cup\{k_1\})=(1+\rho_{\text{temp}})e_{k_1}<\tau+(1+\rho_{\text{temp}})e_{k_1}$, so the if-condition at Line~13 is not satisfied.
Thus, Line~16 is reachable and $S_i\neq\emptyset$.
Next, we have $(1+\rho_i)e(S_i)=l_i$ and each $k\in S_i$ satisfies $l_k\leq l_i$ (guaranteed by Line~6 of the algorithm), so the left part of Eqn.~\eqref{eq:NE_matching} in Proposition~\ref{prop:NE_matching} is satisfied.
The if-statement at Line~13 ensures that the right part of Eqn.~\eqref{eq:NE_matching} in Proposition~\ref{prop:NE_matching} is satisfied for the last agent $k_j$ added to $S_i$.
Since agents in $K_i$ are ordered with decreasing endowments and the term $\tau+(1+\rho)e_i$ on the right-hand side of Eqn.~\eqref{eq:NE_matching} is increasing in $e_i$, the right part of Eqn.~\eqref{eq:NE_matching} in Proposition~\ref{prop:NE_matching} should be satisfied for all agents in $S_i$.
Thus, $S_i$ is a Nash equilibrium given rate of matching funds $\rho_i$.
It remains to prove the two approximation guarantees in the theorem.

Suppose agents are ordered such that $l_1\leq\cdots\leq l_n$.
Let $(S,\delta)$ be the solution output by the algorithm. 
Let $(S_i,\delta_i)$ be the solution output at $i$-th iteration of the for-loop at Line~4.
Let $(S^\ast,\delta^\ast)$ be the optimal solution.

First of all, if the ``break'' statement at Line~10 has ever been executed, we show that the two approximation guarantees in the theorem are satisfied.
Suppose, at the $i$-th iteration of the for-loop at Line~4, the for-loop at Line~8 exits at Line~10.
We know $e(S_i\cup\{k_j\})>l_i$. 
This means adding $k_j$ to $S_i$ would make the total investment more than $l_i$.
We have $l_i-e(S_i)<e_{k_j}$ and $l_i=(1+\rho_i)e(S_i)$.
Therefore, $\rho_i\cdot e(S_i)=l_i-e(S_i)<e_{k_j}<e$.
Moreover, since each agent in $S_i$ has endowment at least $k_j$ by Line~7, we have $\rho_i<\frac{e_{k_j}}{e(S_i)}\leq1\leq\bar{\rho}$.
The solution $(S_i,\rho_i)$ will not be filtered at Line~19-22, and both approximation guarantees, $\rho\cdot e(S)\leq\max\{e,\rho^\ast\cdot e(S^\ast)\}$ and $\rho\leq\max\{1,\rho^\ast\}$, are satisfied.
We will assume the ``break'' statement at Line~10 has never been executed from now on.

Next, we consider the optimal solution $(S^\ast, \rho^\ast)$.
Let $i^\dag$ be the agent with the largest index such that $(1+\rho^\ast)e(S^\ast)\geq l_{i^\dag}$.
Each agent $k\in S^\ast$ must have $l_k\leq l_{i^\dag}$, for otherwise $i^\dag$ should not have been the agent with the largest index.
Therefore, $S^\ast\subseteq K_{i^\dag}$.
The next two paragraphs show that 
$$(1+\rho^\ast)e(S^\ast)=l_{i^\dag}.$$

Suppose otherwise $(1+\rho^\ast)e(S^\ast)>l_{i^\dag}$.
We have $\rho^\ast=0$, for otherwise we can always decrease $\rho^\ast$ a little bit and still guaranteeing $S^\ast$ is a Nash equilibrium (notice that we have $l_{i^\dag}<(1+\rho^\ast)e(S^\ast)<\tau+(1+\rho^\ast)e_i$ for each $i\in S^\ast$, and all the inequalities are strict).
Therefore, we have $\rho^\ast=0$ and $e(S^\ast)>l_{i^\dag}$ now.
By Proposition~\ref{prop:NE_matching}, each $i\in e(S^\ast)$ must satisfy $e(S^\ast)<\tau+e_i$, which implies $l_{i^\dag}<\tau+e_i$.

Let $k_m$ be the agent in $S^\ast$ with minimum endowment.
We have $l_{i^\dag}<\tau+e_{k_m}$.
Since we have assumed the ``break'' statement at Line~10 has never been executed, the for-loop at Line~8 can either terminate when all the element of $K_{i^\dag}$ are added to $S$ or exit at Line~14.
Since $l_{i^\dag}<\tau+e_{k_m}\leq\tau+(1+\rho_{\text{temp}})e_{k_m}\leq\tau+(1+\rho_{\text{temp}})e_{k_j}$ for all $\rho_{\text{temp}}\geq0$ and all $k_j\in K_{i^\dag}$ before $k_m$, we know that the for-loop is still going on by the time $k_m$ is added to $S_{i^\dag}$.
Therefore, $S^\ast\subseteq S_{i^\dag}$.
However, this would imply $e(S_{i^\dag})\geq e(S^\ast)>l_{i^\dag}$, which contradicts to our assumption that the ``break'' statement at Line~10 has not been executed.
We have proved $(1+\rho^\ast)e(S^\ast)=l_{i^\dag}$.

Finally, we show that our assumption that the ``break'' statement at Line~10 has never been executed implies the solution output by our algorithm is optimal.
We first notice that $\rho_{\text{temp}}$ is decreasing after each iteration of the for-loop at Line~8, because $e(S_{i^\dag})$ is increasing when more and more elements are added to $S_{i^\dag}$.
If we have $\rho_{\text{temp}}<\rho^\ast$ at certain iteration, we have $\rho_{i^\dag}<\rho^\ast$. Since $(1+\rho_{i^\dag})e(S_{i^\dag})=l_{i^\dag}=(1+\rho^\ast)e(S^\ast)$, this would imply $e(S_{i^\dag})>e(S^\ast)$, which further implies $\rho_{i^\dag}\cdot e(S_{i^\dag})=l_{i^\dag}-e(S_{i^\dag})<l_{i^\dag}-e(S^\ast)=\rho^\ast\cdot e(S^\ast)$, contradicting to that $(S^\ast,\rho^\ast)$ is optimal.
Therefore, $\rho_{\text{temp}}\geq\rho^\ast$ throughout the for-loop at Line~8.

Let $k_m$ be the agent in $S^\ast$ with minimum endowment again.
Since $l_{i^\dag}<\tau+(1+\rho^\ast)e_{k_m}\leq\tau+(1+\rho_{\text{temp}})e_{k_m}$ (otherwise, the optimal solution $(S^\ast,\rho^\ast)$ should not have been valid), we know that the for-loop at Line~8 is still going on by the time $k_m$ is added to $S_{i^\dag}$.
As a result, $S^\ast\subseteq S_{i^\dag}$ and $e(S^\ast)\leq e(S_{i^\dag})$.
This implies $\rho_{i^\dag}\cdot e(S_{i^\dag})=l_{i^\dag}-e(S_{i^\dag})\leq l_{i^\dag}-e(S^\ast)=\rho^\ast\cdot e(S^\ast)$, which further implies $\rho_{i^\dag}\leq\rho^\ast\cdot\frac{e(S^\ast)}{e(S_{i^\dag})}\leq\rho^\ast$.
Thus, $(S_{i^\dag},\rho_{i^\dag})$ is optimal in terms of both objectives, and so is the final output $(S,\rho)$ of Algorithm~\ref{alg:matchingFunds}.
\end{proof}

\begin{remark}
In Theorem~\ref{thm:approximation_matchingfunds}, we set $\rho^\ast$ be the optimal solution in terms of the \emph{cost} $\rho^\ast\cdot e(S^\ast)$.
If we consider $\rho^\ast$ be the optimal solution in terms of the minimum \emph{rate}, then Algorithm~\ref{alg:matchingFunds} can still obtain guarantee $\rho\leq\max\{1,\rho^\ast\}$, with Line~19 modified to ``find $i^o=\argmin_i\rho_i$''.
 The proof of this is similar to the proof of Theorem~\ref{thm:approximation_matchingfunds} and is left to the readers.
 
 However, we note that the optimal solution in terms of the cost is not always identical to the optimal solution in terms of the rate.
 Consider Example~\ref{ex:big}.
 The solution with optimal cost is $(S=\{1,2,3\},\rho=\frac23)$ and the solution with optimal rate is $(S=\{1,4\},\rho=\frac{7}{11})$.
\end{remark}

Similarly to adding an external investment (\Cref{ex:harm_external}), matching fund can also hurt cooperation.
\begin{example}[Harm of matching fund]
Let $\tau = 9$, $\vec{e} = (4,4,4)$ and $\vec{m} = (4/11, 4/11, 4/11)$.  $S=\{1,2,3\}$ is a cooperative Nash equilibrium.
However, if we use matching fund with $\rho = 1/4$, by Proposition~\ref{prop:NE_matching}, there is no cooperative Nash equilibrium.
\end{example}

\section{Discussion and Conclusion}
In the paper, we study a variant of public goods games, the pivotal participation game, and examine when cooperation can happen in the game.  
First, we show that it is NP-complete to decide the existence of cooperation.  Then we consider two interventions to help agents to cooperate: external investment and matching fund.  We propose two algorithms that can help agents to cooperate with near minimum cost.  Finally, we show that matching fund can be more powerful than an external investment.

There are several interesting directions for future work.  The first is considering more fine-grained interventions.  For instance, we may change each agent's endowment or reward level by taxes or subsidies.  However, these interventions require more information, and the objective is not as well-defined. 
Alternatively, we can consider the information aspect of the interventions.  In the paper, we know the game's parameter and want to design algorithms with reasonable computational complexity.  Those parameters may be unknown and private to agents, and agents may not report them truthfully. However, we believe the computational issue is fundamental and will also arise in the unknown parameter setting.  Finally, some open problems are related to our approximation algorithms, e.g., hardness results about additive approximation for the optimal external investment or matching fund.  Our current construction does not directly extend to these settings.

\bibliographystyle{plainnat}
\bibliography{ijcai21}

\end{document}